\documentclass{article}

\usepackage{amsmath}
\usepackage[psamsfonts]{amssymb}
\usepackage{amsthm}

\newtheorem{Theorem}{Theorem}
\newtheorem{Proposition}{Proposition}
\newtheorem{Lemma}{Lemma}

\newtheorem{Conjecture}{Conjecture}

\title{On the hyperbolicity and causality of the relativistic Euler system under the kinetic equation of state
}

\author{Juan Calvo\footnote{Universitat Pompeu Fabra, Dept. de Tecnologies de la Informaci\'o i les Comunicacions, C/T\`anger 122-140, 08018 Barcelona, Spain. Email adress: juan.calvo@upf.edu.}}

\date{}
\begin{document}
 \maketitle

\begin{abstract}
We show that a pair of conjectures raised in \cite{speck-strain} concerning the construction of normal solutions to the relativistic Boltzmann equation are valid. This ensures that the results in \cite{speck-strain} hold for any range of positive temperatures and that the relativistic Euler system under the kinetic equation of state is hyperbolic and the speed of sound cannot overcome $c/\sqrt{3}$.
\end{abstract}
\section{Introduction}
The purpose of \cite{speck-strain} was to construct solutions to the special relativistic Boltzmann equation as corrections of local relativistic Maxwellians, 
where the thermodynamic fields evolve in time according to the relativistic Euler fluid equations.

To close the relativistic Euler system we add an equation of state. It turns out that in the context of \cite{speck-strain} the right choice is obtained extrapolating the relations that hold between the thermodynamic quantities of the global equilibria of the relativistic Boltzmann equation (known as J\"uttner equilibria). We may call this the kinetic equation of state. 
It is given by
\begin{equation}
\label{state1}
p=k_B n \theta = m_0 c^2 \frac{n}{\beta},
\end{equation}
\begin{equation}
\label{state2}
\rho= m_0 c^2 n \frac{K_1(\beta)}{K_2(\beta)}+3 p,
\end{equation}
\begin{equation}
\label{state3}
n = 4 \pi e^4 m_0^3 c^3 h^{-3}\exp \left( \frac{-\eta}{k_B}\right) \frac{K_2(\beta)}{\beta}\exp \left( \beta \frac{K_1(\beta)}{K_2(\beta)}\right).
\end{equation}
Here $n$ is the proper number density, $\theta$ is the temperature, $\beta= m_0 c^2/(k_B\theta)$ is a dimensionless quantity proportional to the inverse temperature --where $c$ is the speed of light, $m_0$ is the mass of gas particles and $k_B$ is Boltzmann's constant--, $p$ stands for the pressure, $\rho$ stands for proper energy density and $\eta$ stands for the entropy per particle, while $h$ is Planck's constant and $K_1,\ K_2$ are modified Bessel functions (which we define below). For the related computations see \cite{CercignaniKremer,deGroot,Synge} for instance. In fact,  it can be checked that the same equation of state arises in the formal hydrodynamic limit of
 a kinetic model having J\"uttner distributions as equilibria. This is the case for BGK-type models, see \cite{BGK,AW,CercignaniKremer,Majorana, Marle2} for example.

During their analysis, the authors of \cite{speck-strain} ran into some difficulties as several properties of the kinetic equation of state (local solvability for any of the thermodynamic variables in terms of any other two of them, hyperbolicity and causality of the relativistic Euler system under the kinetic equation of state) were required but no rigorous proof for them seemed to be available in the literature. These properties were proven in \cite{speck-strain} to hold true in certain temperature regimes and conjectured to be true in general. Then a pair of  statements were raised in \cite{speck-strain}, which can be quoted as follows:
\begin{Conjecture}
The map $(n,\beta) \leftrightarrow (\eta,\rho)$ given implicitly by the kinetic equation of state (\ref{state1})--(\ref{state3}) is an auto-diffeomorphism of $]0,\infty[\times]0,\infty[$.
\end{Conjecture}
\begin{Conjecture}
Under the kinetic equation of state (\ref{state1})--(\ref{state3}), there holds that:
\begin{itemize}

\item There exists a smooth function $f_{\rm kinetic}$ such that $p$ can be expressed in terms of $\eta$ and $\rho$ as\footnote{Let us mention that we use the term ``kinetic equation of state'' in a broader sense than what is done in \cite{speck-strain}, where this term is used to refer to $f_{\rm kinetic}$, not to \eqref{state1}--\eqref{state3}.} $p=f_{\rm kinetic}(\eta,\rho)$.  

\item The relativistic Euler system is hyperbolic.

\item The relativistic Euler system is causal (the speed of sound $c_S:=c\sqrt{\frac{\partial p}{\partial \rho}|_{\eta}}$ is  real and less than the speed of light). Furthermore, $ 0<c_S<c/\sqrt{3}$.
\end{itemize} 
\end{Conjecture}
It was also shown in \cite{speck-strain} that these conjectures can be phrased in terms of relations between modified Bessel functions as 
\begin{equation}
\label{conjetura1}
3\frac{K_1(\beta)}{K_2(\beta)}+\beta \left(\frac{K_1(\beta)}{K_2(\beta)} \right)^2-\beta-\frac{4}{\beta}<0,\quad \forall \beta>0
\end{equation}
for the first conjecture and
\begin{equation}
\label{conjetura2}
3<3+\beta\frac{K_1(\beta)}{K_2(\beta)}+\frac{4\frac{K_1(\beta)}{K_2(\beta)} + \beta\left(\frac{K_1(\beta)}{K_2(\beta)} \right)^2- \beta}{3\frac{K_1(\beta)}{K_2(\beta)} + \beta\left(\frac{K_1(\beta)}{K_2(\beta)} \right)^2- \beta-\frac{4}{\beta}}<+\infty ,\quad \forall \beta>0
\end{equation}
for the second conjecture. The purpose of this paper is to show that the previous inequalities \eqref{conjetura1}--\eqref{conjetura2} indeed hold true, thus Conjecture 1 and Conjecture 2 are then shown to be valid. This implies in particular that the results in \cite{speck-strain} concerning the hydrodynamic limit of the relativistic Boltzmann equation hold true for any range of positive temperatures. The bound on the sound speed may also be important in cosmology: some results suggest that $c/\sqrt{3}$ may be a boundary value separating unstable fluid regimes from stable fluid regimes for the case of nearly-uniform solutions to the fluid equations in rapidly expanding spacetimes \cite{Rendall04, Speck2012}. Specifically, an equation of state of the form $p=c_S^2\rho$ was addressed in these papers. On one hand, evidence was found pointing that solutions may be unstable in the regime given by $c_S>c/\sqrt{3}$. On the other hand, global future stability for small data solutions when $0\le c_S \le c/\sqrt{3}$ and the spacetime is expanding at a sufficiently rapid rate has been shown in \cite{Speck2012}.


\section{Proofs of the conjectures}
Modified Bessel functions can be defined as \cite{librotablas,Synge} 
\begin{equation}
\label{ch-representation}
K_j(\beta) = \int_0^\infty \cosh (jr) e^{-\beta \cosh(r)}\ dr \ge 0, \quad j \in \mathbb{N}, \ \beta>0.
\end{equation}
We recall that these functions obey the following recurrence relation:
\begin{equation}
\label{recurrence}
K_{j+1}(\beta)= \frac{2j}{\beta}K_j(\beta) + K_{j-1}(\beta), \quad j \in \mathbb{N}, \ \beta>0.
\end{equation}
We are now ready to display our first result.
\begin{Theorem}
\label{uno}
Inequality \eqref{conjetura1} holds true for any $\beta>0$. As a consequence, Conjecture 1 holds true.
\end{Theorem}
This result follows at once from the following non-trivial fact (see \cite{Kunik2004}, Lemma 2.2 for instance):
\begin{Proposition}
\label{parte1}
Let $K_1$ and $K_2$ be the modified Bessel functions defined by \eqref{ch-representation}. Then 
\begin{equation}
\label{repform}
\left(\frac{K_1(\beta)}{K_2(\beta)} \right)'=\left(\frac{K_1(\beta)}{K_2(\beta)} \right)^2 +\frac{3}{\beta}\frac{K_1(\beta)}{K_2(\beta)} -1
\end{equation}
and the following inequality holds true:
\begin{equation}
\label{kunik}
 \left(\frac{K_1(\beta)}{K_2(\beta)} \right)'<\frac{3}{\beta^2}, \quad \forall \beta>0.
\end{equation}
\end{Proposition}
A proof for Proposition \ref{parte1} can be found in \cite{Kunik2004}. See also \cite{Synge}, p. 89. Inequality \eqref{kunik} shows that the specific energy of a gas in equilibrium (defined as $\psi(\beta)=\frac{3}{\beta}+\frac{K_1(\beta)}{K_2(\beta)}= \frac{\rho}{n}$) is an increasing function of the temperature. 

The rest of this section is devoted to prove the following statement:
\begin{Theorem}
\label{dos}
Inequality \eqref{conjetura2} holds true for any $\beta>0$. As a consequence, Conjecture 2 holds true.
\end{Theorem}
More precisely, the middle term in \eqref{conjetura2} equals $(\frac{\partial p}{\partial \rho}|_{\eta})^{-1}$. Then when the second inequality in \eqref{conjetura2} is valid we get the existence of $f_{\rm kinetic}$, the hyperbolicity of the relativistic Euler system and the fact that $c_S>0$ all at once. When the first inequality in \eqref{conjetura2} holds we get the bound $c_S<c/\sqrt{3}$. See \cite{speck-strain} for details.

Prior to the proof of Theorem \ref{dos}, let us stress that the second inequality in (\ref{conjetura2}) is trivially satisfied once we have proved Theorem \ref{uno}, as the denominator in \eqref{conjetura2} never vanishes. Thus, it suffices to focus on the first inequality. Then, the first thing we do to prove the second conjecture is to rephrase inequality \eqref{conjetura2} as an inequivalent inequality which is more suitable to use certain bounds on the ratio $K_1/K_2$. We now state and prove this reformulation in the following result.

\begin{Lemma}
To show that the first inequality in \eqref{conjetura2} is satisfied for every $\beta>0$ is equivalent to show that
\begin{equation}
\label{reformulation}
\left( \frac{K_1(\beta)}{K_2(\beta)}\right)^2 < 1- \frac{3}{4 + \beta \frac{K_1(\beta)}{K_2(\beta)}}
\end{equation}
holds for every $\beta>0$.
\end{Lemma}

\begin{proof}
We start using \eqref{repform} to write \eqref{conjetura2} as
\begin{equation}
\label{primera}
 0< \beta\frac{K_1(\beta)}{K_2(\beta)} + \frac{\frac{K_1(\beta)}{K_2(\beta)}+\beta\left( \frac{K_1(\beta)}{K_2(\beta)}\right)'}{\beta\left( \frac{K_1(\beta)}{K_2(\beta)}\right)'-\frac{4}{\beta}}\quad \forall \beta>0.
\end{equation}
Note that thanks to \eqref{kunik} we can ensure that the denominator of the second term on the rhs is strictly negative. Due to this fact, \eqref{primera} is equivalent to
\begin{equation}
\label{segunda}
\frac{K_1(\beta)}{K_2(\beta)} + \beta\left( \frac{K_1(\beta)}{K_2(\beta)}\right)'<-\beta^2 \frac{K_1(\beta)}{K_2(\beta)} \left( \frac{K_1(\beta)}{K_2(\beta)}\right)' + 4 \frac{K_1(\beta)}{K_2(\beta)}
\end{equation}
We substitute now the value of the derivative of the ratio $K_1/K_2$ given by \eqref{repform} in \eqref{segunda}, obtaining
\begin{equation}
\label{tercera}
\left(  \beta + \beta^2 \frac{K_1(\beta)}{K_2(\beta)} \right) \left\{ \left( \frac{K_1(\beta)}{K_2(\beta)}\right)^2 +\frac{3}{\beta}\frac{K_1(\beta)}{K_2(\beta)}-1\right\} < 3 \frac{K_1(\beta)}{K_2(\beta)}.
\end{equation}
Expanding the product and rearranging a bit, inequality \eqref{tercera} can be recast as
\begin{equation}
\label{cuarta}
\beta \left( \frac{K_1(\beta)}{K_2(\beta)}\right)^2 \left(4 + \beta \frac{K_1(\beta)}{K_2(\beta)} \right)<\beta+\beta^2\frac{K_1(\beta)}{K_2(\beta)}.
\end{equation}
We divide by $\beta$ and the third term on the lhs of \eqref{cuarta} to arrive to \eqref{reformulation}.
\end{proof}
To begin with, we can show that \eqref{reformulation} holds for $0<\beta<1$. This is a consequence of the following estimate.
\begin{Lemma}
\label{ldos}
The inequality
$
0\le K_1/K_2\le \beta/2
$
holds for all $\beta>0$.
\end{Lemma}
\begin{proof}
Using \eqref{recurrence} for $j=1$ and the fact that $K_0\ge 0$ we deduce that $K_2(\beta)\ge 2 K_1(\beta)/\beta$ for all $\beta>0$. As $K_j(\beta)>0$ for all $\beta>0$ and $j \in \mathbb{N}$ the result follows. 
\end{proof}
Now we replace the upper estimate granted by Lemma \ref{ldos} on the lhs of \eqref{reformulation} and we also replace the lower estimate on its rhs, so our claim follows. To proceed further we use again the recurrence relation \eqref{recurrence} to obtain more estimates on $K_1/K_2$. 
We display a procedure which is highly reminiscent of what is done in \cite{Synge} to obtain estimates on the functions $K_j(\beta)$ by means of their asymptotic expansions for $\beta \gg 1$ (see also \cite{librotablas} p. 378).
\begin{Lemma}
\label{cuatro}
The following estimates hold for $\beta >1/2$:
\begin{equation}
\label{septima}
\frac{128 \beta^3+48 \beta^2-15 \beta}{128 \beta^3 +240 \beta^2 +105 \beta -30} \le \frac{K_1(\beta)}{K_2(\beta)}\le \frac{4 \beta^2+3 \beta/2}{4 \beta^2 +15 \beta/2+3} \quad \forall \beta\ge 1/2.
\end{equation}
\end{Lemma}
\begin{proof}
We shall rewrite formula \eqref{ch-representation} by means of the change of variables $\frac{z^2}{4\beta}=\sinh^2(r/2)$. In such a way,
\begin{equation}
\label{cacero}
K_0(\beta)=\frac{e^{-\beta}}{\sqrt{\beta}} \int_0^\infty \frac{1}{\sqrt{1+\frac{z^2}{4\beta}}} e^{-z^2/2}\ dz
\end{equation}
and 
\begin{equation}
\label{casuma}
K_0(\beta)+K_1(\beta)=2 \frac{e^{-\beta}}{\sqrt{\beta}} \int_0^\infty
\sqrt{1+\frac{z^2}{4\beta}} \ e^{-z^2/2}\ dz.
\end{equation}
Now we estimate the integrands. A full binomial expansion for the square roots in \eqref{cacero}--\eqref{casuma} would yield asymptotic expansions for $K_0$ and $K_0+K_1$. Here we will content ourselves keeping just a few terms, which will be enough to prove Theorem \ref{dos}. Following \cite{Synge}, we recall that Taylor's theorem yields the following representation:
\begin{equation}
\label{binomial}
(1+x)^{n-\frac{1}{2}} = 1+\sum_{i=1}^{p-1} {n-\frac{1}{2} \choose i} x^i+ p {n-\frac{1}{2} \choose p} I_p(x), \quad p=1,2,3,\ldots
\end{equation}
where
\begin{equation}
\label{remains}
I_p(x)= \int_0^1 (1-t)^{p-1} (1+t x)^{n-\frac{1}{2}-p} x^p\ dt.
\end{equation}
We notice that $I_p(x)>0$ for $x>0$ and then the remainder term has the same sign as ${n-1/2 \choose p} x^p$ does, which is the first term of the binomial series that we have discarded. If we cut the expansion in a negative term we get an estimate from below, while we get an estimate from above if we cut the expansion in a positive term. Thus, if we substitute $x$ by $x^2$ in \eqref{binomial}--\eqref{remains} we obtain that
\begin{equation}
\label{Taylorbounds}
 1-\frac{x^2}{2} \le \frac{1}{\sqrt{1+x^2}}\le 1-\frac{x^2}{2}+\frac{3x^4}{8},\ \forall x>0,
\end{equation}
\begin{equation}
\label{Taylorbounds2}
 1+\frac{x^2}{2}-\frac{x^4}{8} \le \sqrt{1+x^2}\le 1+\frac{x^2}{2},\ \forall x>0.
\end{equation}
We plug \eqref{Taylorbounds}--\eqref{Taylorbounds2} into \eqref{cacero}--\eqref{casuma} to obtain for all $\beta >0$ the following bounds:
\begin{equation}
\label{quinta}
\sqrt{\frac{\pi}{2}} \frac{e^{-\beta}}{\sqrt{\beta}} \left( 1-\frac{1}{8 \beta}\right) \le K_0(\beta)\le  \sqrt{\frac{\pi}{2}} \frac{e^{-\beta}}{\sqrt{\beta}} \left( 1-\frac{1}{8 \beta}+\frac{9}{128 \beta^2}\right)
\end{equation}
and
\begin{equation}
\label{sexta}
2\sqrt{\frac{\pi}{2}} \frac{e^{-\beta}}{\sqrt{\beta}} \left(1+\frac{1}{8 \beta}- \frac{3}{128 \beta^2}\right) \le K_0(\beta)+K_1(\beta)\le 2\sqrt{\frac{\pi}{2}} \frac{e^{-\beta}}{\sqrt{\beta}} \left(1+\frac{1}{8 \beta}\right). 
\end{equation}
Note that the lower estimates that we get in such a way are positive for (say) $\beta \ge 1/2$. To conclude, we note that thanks to \eqref{recurrence} we have $K_2/K_1=2/\beta + K_0/K_1$; we rewrite $K_1/K_0$ as $(K_0 + K_1)/K_0 -1$ and use \eqref{quinta}--\eqref{sexta} to obtain \eqref{septima}.
\end{proof}

Substitution of the bounds given by \eqref{septima} into \eqref{reformulation} shows that, if 
\begin{equation}
\label{septima2}
\left( \frac{4 \beta^2+3 \beta/2}{4 \beta^2 +15 \beta/2+3}\right)^2 < 1- \frac{3}{4+\beta \frac{128 \beta^3+48 \beta^2-15 \beta}{128 \beta^3 +240 \beta^2 +105 \beta -30}}
\end{equation}
holds true for $\beta\ge 1/2$, then the inequality \eqref{reformulation} holds for $\beta \ge 1/2$ and we will be done. 
After some computations we see that \eqref{septima2} is equivalent to 
\begin{equation}
\label{octava}
\frac{16 \beta^4+12 \beta^3+ 9 \beta^2/4}{
16 \beta^4 + 60 \beta^3 + 321 \beta^2/4+45 \beta + 9}  <\frac{128 \beta^4 +176 \beta^3+225 \beta^2+105 \beta-30}{128 \beta^4+560 \beta^3+945 \beta^2 + 420 \beta -120}.
\end{equation}
Both numerators and denominators in \eqref{octava}  are strictly positive in the range $\beta\ge 1/2$, so we can multiply by the denominators in \eqref{octava}. Thus \eqref{octava} is finally equivalent to
\begin{equation}
\label{novena}
\frac{3}{4}\left(3072 \beta^6 + 20992 \beta^5 +36936 \beta^4 + 25107 \beta^3 +6150 \beta^2 -540 \beta -360 \right) >0.
\end{equation}
Then \eqref{novena} is easily seen to be true for $\beta \ge 1/2$, in fact it suffices to keep the last three terms to ensure it. Thus we have proved Theorem \ref{dos}.
\section*{Acknowledgments} I would like to thank Jared Speck and Robert M. Strain for their support and their most useful comments on a first draft of this document. I also thank the  referees of Commun. Pure Appl. Anal. for their most valuable comments and suggestions. I am  grateful to Juan Soler for useful discussions about the contents of this paper.




\begin{thebibliography}{19}

\bibitem{librotablas} 
     \newblock M. Abramovitz and I. A. Stegun,
     \newblock ``Handbook of Mathematical Functions with Formulas, Graphs, and Mathematical Tables,"
     \newblock New York: Dover Publications, 1972.
     
\bibitem{AW} 
    \newblock J.L. Anderson and H.R. Witting, 
        \newblock \emph{A relativistic relaxation-time model for the Boltzmann equation}, Physica, \textbf{74} (1974), 466--488.
             
\bibitem{BGK}
 \newblock A. Bellouquid, J. Calvo, J. Nieto and J. Soler, 
 \newblock \emph{On the relativistic BGK-Boltzmann model: asymptotics and hydrodynamics} (arXiv: 1205.6603).

\bibitem{CercignaniKremer} 
     \newblock C. Cercignani and G. Medeiros Kremer,
     \newblock ``The Relativistic Boltzmann Equation: Theory and Applications,"
     \newblock Birkh\"auser, Berlin, 2003.
%
\bibitem{deGroot} 
     \newblock S.R. De Groot, W. A. van Leuwen and Ch. G. van Weert,
     \newblock ``Relativistic Kinetic Theory. Principles and Applications,"
     \newblock North Holland, Amsterdam, 1980.

\bibitem{Kunik2004}
    \newblock M. Kunik, S. Qamar and G. Warnecke, 
    \newblock \emph{Kinetic schemes for the relativistic gas dynamics}, Numer. Math., \textbf{97} (2004), 159--191.

\bibitem{Majorana} 
    \newblock A. Majorana, 
        \newblock \emph{Relativistic relaxation models for a simple gas}, J. Math. Phys., \textbf{31} (1990), 2042--2046.
%

\bibitem{Marle2} 
    \newblock C. Marle, 
    \newblock \emph{Sur l'\'etablissement des \'equations de l'hydrodynamique des fluides relativistes dissipatifs.I.-- L'equation de Boltzmann relativiste},
   \newblock (French), Ann. Inst. H. Poincar\'e, \textbf{10} (1969), 67--127.
   
 \bibitem{Rendall04} 
    \newblock A.D. Rendall, 
        \newblock \emph{Asymptotics of solutions of the Einstein equations with positive cosmological constant}, Ann. Henri Poincar\'e, \textbf{5} (2004), 1041--1064.  
        
 \bibitem{Speck2012}
 \newblock J. Speck,       
 \newblock \emph{The stabilizing effect of spacetime expansion on relativistic fluids with sharp results for the radiation equation of state} (ArXiv 1201.1963).

\bibitem{speck-strain} 
    \newblock J. Speck and R. M. Strain, 
    \newblock \emph{Hilbert expansion from the Boltzmann equation to relativistic fluids}, Comm. Math. Phys., \textbf{304} (2011), 229--280.

 
\bibitem{Synge} 
     \newblock J. L. Synge, 
          \newblock ``The Relativistic Gas,"
     \newblock North Holland, Amsterdam, 1957.
     
\end{thebibliography}
 \end{document}